\documentclass[review,onefignum,onetabnum]{siamart190516}

\usepackage{lipsum}
\usepackage{amsfonts}
\usepackage{graphicx}
\usepackage{epstopdf}
\usepackage{algorithmic}
\usepackage{tcolorbox}
\ifpdf
  \DeclareGraphicsExtensions{.eps,.pdf,.png,.jpg}
\else
  \DeclareGraphicsExtensions{.eps}
\fi

\newsiamremark{rem}{Remark}
\newsiamremark{hypothesis}{Hypothesis}
\newsiamremark{assumption}{Assumption}
\newtheorem{lem}{Lemma}[section]
\newtheorem{prop}{Proposition}[section]
\crefname{hypothesis}{Hypothesis}{Hypotheses}
\newsiamthm{claim}{Claim}

\headers{An Example Article}{A. Ibragimov and A. Peace}
\nolinenumbers
\title{Light driven interactions in spatial predator-prey model with toxicant chemotaxis
 \thanks{Submitted to the editors DATE.
\funding{This work was partially funded by NSF DMS-1615697}}}

\author{Akif Ibragimov\thanks{Department of Mathematics and Statistics, Texas Tech University, Lubbock, TX
  (\email{Akif.Ibraguimov@ttu.edu}).}
\and Angela Peace\thanks{Department of Mathematics and Statistics, Texas Tech University, Lubbock, TX
  (\email{A.Peace@ttu.edu}).}
}

\usepackage{amsopn}

\begin{document}

\maketitle

\begin{abstract}
We develop and analyze a spatial temporal model of light driven ecotoxicological processes, motivated by an aquatic predator-prey system of algae and \textsl{Daphnia} subject to a contaminant. Population dynamics are driven by light, which is periodic in time and varies with spatial depth. The existence and uniqueness of spatial and temporal dependent periodic solutions are shown and analytical functions of the solutions under parameter constraints are presented.  We conduct Turing stability analyses of solutions with respect to perturbations of initial conditions. Given a perturbation to a periodic equilibrium state, we show the system will return to this equilibrium state as long as motility is fast enough and/or the reservoir depth is shallow enough. Analytical results assume some Dirichlet boundary conditions that match the periodic equilibrium state, however numerical simulations with more relaxed boundary conditions capture similar periodic solutions. The work sheds light onto spatially dependent population dynamics that are driven by periodic forces, such as light levels. 
\end{abstract}

\begin{keywords}
periodic light, ecotoxicology, population dynamics, Turing stability
\end{keywords}

\begin{AMS}
  35A01, 35A02, 35A08, 92B05 
\end{AMS}

\section{Introduction}

Ecotoxicological modeling aims to better understand the fate of chemicals and their influence on population dynamics \cite{jorgensen2013modelling, huang2013model, peace2016somatic}. 
Spatial dynamics can play an important role in population dynamics \cite{rana2018mechanistically} and predicting how contaminants cycle through aquatic food webs, therefore spatially explicit modeling approaches have been proposed \cite{caswell1996demography, chaumot2003ecotoxicology}.  Additionally, temporal variations in light levels impacts organismal light history traits and behaviors.  Empirical evidence suggests that metabolic activity,  is dependent on diel cycles \cite{hu2018hard}.  Indeed, Ribalet et al. 2015 \cite{ribalet2015light} found that the cell production and mortality rates of aquatic primary producers are tightly synchronized to day/night cycles, where high growth but low mortality rates are observed during day light and low growth but high mortality rates are observed during the night. 

Here, we develop and analyze a spatially explicit model of an aquatic predator-prey system subject to contaminants with periodic light-driven population dynamics. We expand previous modeling efforts that assume constant light levels and neglect spatial dynamics \cite{huang2014development}. The extended model is a non-smooth periodic system of nonlinear partial differential equations. While our predator-prey system is generalizable we develop it for algae, a primary producer as the prey, and \textsl{Daphnia} a grazer for the predator.  In order to incorporate spatial movements, we assume dissolved toxicants follow diffusion, and the populations follow motility and chemotaxis.  Algae movement is driven by light availability, with light-oriented motility and \textsl{Dapnia} movement is driven by motility and chemotaxis towards their prey, as well as away from dissolved toxicants. For simplicity, we choose a simple periodic piecewise function for deil cycles. 

After formulating a full model, we make a number of simplifying assumptions to facilitate analytical model analyses. The results is a simplified three dimensional system of non-smooth periodic partial differential equations.  We identify depth-dependent and time-periodic solutions which we called the equilibrium state.  Assuming no flux boundary conditions for the sides of the domain, and Dirichlet boundary conditions on the top and bottom that match this interior time-periodic equilibrium state, we prove the existence and uniqueness of this bounded solution. We present analytical representations of the solutions satisfying given constraints on the parameters and provide stability analyses of these with respect to perturbations.  The presented Turing stability analyses is based on the Carath{\'e}odory principle. These analytical techniques are similar to our previous papers \cite{ibragimov2010stability, ibragimov2010stability2, yang2013computational} where we established stability criterion for biomedical chemotactic processes with distributed parameters for equilibria, which are constant in time and space.

When considering stability with respect to perturbations, we explore the impact of the size of the domain, as well as, the motility capabilities of the algae, \textsl{Daphnia}, and toxicant diffusion.  We show that, given a perturbation, the system will return to the time-periodic equilibrium state as long as motility  coefficients are large enough and/or the depth of the reservoir is small enough.  Furthermore, numerical simulations provide evidence of similar periodic solutions assuming complete no flux boundary conditions, suggesting that the Dirichlet boundary conditions we assumed in the analytical analyses do not drive the existence of these periodic solutions.

\section{Model formulation}
\label{sec:main}
\subsection{Base nonspatial model}
We start with the toxicant-mediated predator-prey model developed by Huang et al. 2014 \cite{huang2014development},  
\begin{subequations}
\begin{align}
\frac{dU}{dt}&=\underbrace{b(C_U,U)U}_{\text{\parbox{2cm}{\centering gain from  \\[-4pt] growth}}} -\underbrace{d_1(C_U,U)U}_{\text{\parbox{2cm}{\centering loss from  \\[-4pt] death}}}-\underbrace{f(U)V}_{\text{\parbox{2cm}{\centering loss from  \\[-4pt] predation}}} \\
\frac{dV}{dt}&=\underbrace{e(C_V, V)f(U)V}_{\text{\parbox{2cm}{\centering gain from  \\[-4pt] growth}}} -\underbrace{d_2(C_V, V)V}_{\text{\parbox{2cm}{\centering loss from  \\[-4pt] death}}}\\
\frac{dC_U}{dt}&=\underbrace{a_1C_DU}_{\text{\parbox{2cm}{\centering toxicant\\[-4pt] uptake}}} -\underbrace{\sigma_1C_U}_{\text{\parbox{1.5cm}{\centering efflux  \\[-4pt] }}}-\underbrace{d_1(C_U, U)C_U}_{\text{\parbox{2cm}{\centering loss due to   \\[-4pt] death}}} -
\underbrace{f(U)V\frac{C_U}{U}}_{\text{\parbox{2cm}{\centering loss from  \\[-4pt] predation}}} \\
\frac{dC_V}{dt}&=\underbrace{a_2C_DV}_{\text{\parbox{2cm}{\centering toxicant\\[-4pt] uptake}}} -\underbrace{\sigma_2C_V}_{\text{\parbox{1.5cm}{\centering efflux  \\[-4pt] }}}-\underbrace{d_2(C_V, V)C_V}_{\text{\parbox{2cm}{\centering loss due to   \\[-4pt] death}}} +\underbrace{\xi f(U)V\frac{C_U}{U}}_{\text{\parbox{2cm}{\centering gain from  \\[-4pt] predation}}}
\end{align}
\label{huangmodel}
\end{subequations}
where  $U$ and $V$ denote the densities of the prey (algae) and predator (\textsl{Daphnia}) populations respectively.   The contaminant is divided into three pools; contaminant in the prey population $C_U$, contaminant in the predator population $C_V$, and contaminant dissolved in the environment $C_D$. Here, $C_D$ is assumed constant. The fractions $\frac{C_U}{U}$  and $\frac{C_V}{V}$  give toxicant body burdens, or the the concentrations of the toxicant in the prey and predatory respectively. Predator functional response $f(U)$ is assumed to be a Holling type II functional response,
\begin{equation}\label{functionalresponse}
 f(U)=\frac{cU}{a+U}
\end{equation}
where $c$ is the predator's maximal ingestion rate and $a$ is the half saturation constant.  Growth terms $b(C_U, U)$ and $e(C_V, V)$,  and death terms $d_1(C_U,C)$, and $d_2(C_V, V)$  are functions that depend on the respective body burdens. The prey growth rate takes the following form,
\begin{equation}\label{b_u} 
b(C_U,U)=\frac{\alpha_1\max \{0,1-\alpha_2 \frac{C_U}{U} \} }{1+\alpha_3U}
\end{equation}
where $\alpha_1$ is the maximal growth rate, $\alpha_2$ is the effect of the toxicant on growth, and $\alpha_3$ accounts for crowding effects. The maximum function  $0\leq \max \{0,1-\alpha_2 \frac{C_U}{U} \} \leq 1$ comes from a linear dose response for the growth \cite{huang2013model, huang2014development}.  The predator growth rate is the product of a functional response with a toxicant-dependent conversion efficiency, $e(C_V, V)f(U)$. Another linear dose response is used to to represent the effect of the toxicant on the predator growth rate via the conversion efficiency, 
\begin{equation}
e(C_V,V)= \hat{e}\max \left\{ 0, 1-\beta_2\frac{C_V}{V} \right\}
\label{reproductionFunction}
\end{equation}
where $\hat{e}$ is the maximal conversion efficiency and $\beta_2$ is the effect of the toxicant on predator's growth. If the predator body burden, $\frac{C_V}{V}$, reaches the threshold $\frac{1}{\beta_2}$ the efficiency is zero and growth ceases. Huang et al. 2014 \cite{huang2014development} used the  power law to represent the relationship between toxicant concentrations organism mortality rate, as recommended by the committee on toxicology of the National Research Council in 1992 and tested in \cite{miller2000haber}. Mortality rates, as functions of the body burdens take the following form
\begin{equation} \label{mortality}
d_1(C_U,U)=h_1\left[\frac{C_U}{U}\right]^{I_1}+m_1 \qquad \text{and} \qquad d_2(C_V,V)=h_2\left[\frac{C_V}{V}\right]^{I_2}+m_2 
\end{equation}
where $m_1, m_2$ are natural mortality rates, and $h_1, h_2$ are positive coefficients and the powers of the body burden expressions are assumed to be one, $I_1=I_2=1$. 

\subsection{Spatial extension} 
Now, we extend the above model to incorporate spatial dynamics in the domain of interest $\Omega.$  
In this manuscript we assume the domain to be a simple rectangular parallelogram,  
$$\Omega=\Omega_0\times \{-U,0\},$$
with $\Omega_0=\{(x_2,x_3)\in \mathbb{R}^2\}.$
We use the following notation for the top boundary: $\Gamma^0=\Omega_0\times\{x_1=0\},$  bottom boundary: $\Gamma^H=\Omega_0\times\{x_1=-H\},$ and side boundary: $\Gamma_2=\partial \Omega_0 \times \{-H,0\}.$
It is customary to set $\Gamma_1=\Gamma^H\cup\Gamma^0.$ Then evidently $\partial \Omega = \Gamma_1\cup\Gamma_2.$

Often for convenience we denote $x_1=z$ for the depth coordinate.  
The spatial movement of the populations are described by motility and chemotaxis. The prey's movement is governed by light availability. We assume light-oriented motility, 
$$ \bigtriangledown(D_U(C_U, U,H)\cdot\bigtriangledown U)$$
with the matrix,
$$D_U (C_U,U,H)=\begin{pmatrix}D_1(C_U,U)&0&0\\0&D_1(C_U,U)&0\\0&0&D_1(H)
\end{pmatrix}$$
where function $D_1(H)$ decreases with $H,$ for example $D_1(H)\sim\frac{1}{1+H}$.  In addition to this diffusion, we assume the prey follows light-oriented locomotion, 
$$-v_0\frac{\partial U}{\partial z},$$
where the locomotion velocity $v_0$ is assumed to be a constant.  Here, $z$ represents the depth coordinate $x_1=z$. 

The predator's movement is determined by motility, as well as chemotaxis towards the prey and away from dissolved toxicants, 
$$\underbrace{\bigtriangledown(D_V(C_V,V)\bigtriangledown V)}_{\text{\parbox{2cm}{\centering Motility}}}-\underbrace{\bigtriangledown\left(\chi_U \frac{V}{U} \bigtriangledown U\right)}_{\text{\parbox{2cm}{\centering chemotaxis  \\[-4pt] towards prey}}}+\underbrace{\bigtriangledown\left(\chi_c V \bigtriangledown C_D\right)}_{\text{\parbox{2cm}{\centering chemotaxis  \\[-4pt] away from tox}}}$$
with constant coefficients $D_V, \chi_U$, and  $\chi_c$. 

Extending to this spatial framework allows us to incorporate the influence of depth-dependent light levels. In addition to variations with depth, we assume that light varies periodically with time due to day-night cycles. We define the function $h(z,t)$ to incorporate the influence of time and depth varying light levels on population dynamics:
\begin{equation}\label{h-def}
h(z,t) = \left\{
        \begin{array}{ll}
           e^{\omega t}\exp(\gamma z)& \quad \ 0\leq t\leq T\\
           e^{\omega(2- t)}\exp(\gamma z) & \quad T \leq t \leq \overline{T}
        \end{array}
    \right.
\end{equation}
where constant $\omega$ determines temporally light variations and constant $\gamma$ determines spatial light variations with depth.  For convenience, we will allow the period of the diel cycles to be equal to $\overline{T}=2$, with switching time in the piecewise function as $T=1$. 

Since the prey is algae, a primary producer using photosynthesis to grow, we modify the prey's growth rate given above in Eq. \eqref{b_u} to be light-dependent and take the following form, 
\begin{equation}\label{newb_u} 
b(C_U,U, h(z,t))=\frac{\alpha_1h(z,t)\max \{0,1-\alpha_2 \frac{C_U}{U} \} }{1+\alpha_3U}. 
\end{equation}
This is now a toxicant (body-burden) and light dependent growth rate. 

We also modify the mortality rates to be dependent on light-levels. It has been shown that natural mortality rates tend to vary with light levels, with low moralities rates occurring during the daytime and higher mortality rates during the night \cite{ribalet2015light}. We modify the natural mortality constants given above in Eq. \eqref{mortality} to be light-dependent functions, 

\begin{subequations}
\begin{align}\label{d_1_2}
d_1(C_U,U, h(z,t))&=h_1\left[\frac{C_U}{U}\right]^{I_1}+m_1(h(z,t)) \\
d_2(C_V,V, h(z,t))&=h_2\left[\frac{C_V}{V}\right]^{I_2}+m_2(h(z,t)).
\end{align}
\label{newmortality}
\end{subequations}
Here,  
\begin{subequations} \label{newmortality_piecewise}
\begin{align}\label{m_1_2}
m_1(h(z,t)) &= \widehat{m}_1  \frac{h(z,t)}{1+\beta h(z,t)}+\mathcal{M}_1(t)\\
m_2(h(z,t)) &= \widehat{m}_2  \frac{h(z,t)}{1+\beta h(z,t)}+\mathcal{M}_2(t),
\end{align}
\end{subequations}
where
\begin{equation}\label{m_1^0-def}
\mathcal{M}_i(t) = \left\{
        \begin{array}{ll}
          \overline{m}_i-\omega ,& \quad 0\leq t\leq 1 ;\\
         \overline{m}_i+\omega ,& \quad 1\leq t\leq 2,
        \end{array}
    \right.
\end{equation}
and $\widehat{m}_1$ and $\widehat{m}_2$ are coefficients for maximum mortalities of the prey and predator respectively, and $\overline{m}_1$ and $\overline{m}_2$ are baseline values with respect to parameter $\omega$, which switches sign in the above expressions depending on the cycle for light availability.  
\begin{rem}\label{omega}
It is important to state that  parameter $\omega$ in the above equation \eqref{m_1^0-def}, may not be exactly the same as it is in the definition of the intensity function $h(z,t)$, equation \eqref{h-def}, but these values are likely related. The biological relationship between these two parameters is left as future work. The mathematical implications of our assumptions here are also left as future work, but we note that these can be investigated using the concept of $\it{differential \ inclusion}$, which was first systematically studied by by A.F. Filippov \cite{filippov2013differential}.  
\end{rem}

Putting all these assumptions together, the full model can be written in the following form: 
\begin{subequations} \label{fullmodel}
\begin{align}
\label{prey-eq}
\frac{\partial U}{\partial t}&=\underbrace{b(C_U,U, h(z,t))U}_{\text{\parbox{2cm}{\centering gain from  \\[-4pt] growth}}} -\underbrace{d_1(C_U,U,h(z,t))U}_{\text{\parbox{2cm}{\centering loss from  \\[-4pt] death}}}-\underbrace{f(U)V}_{\text{\parbox{0.75cm}{\centering pred  \\[-4pt] loss}}}\\ \nonumber &\hspace{5cm}+ \underbrace{\bigtriangledown(D_U(C_U, U, H)\cdot\bigtriangledown U)}_{\text{\parbox{2cm}{\centering Motility}}} \underbrace{-v_0\frac{\partial U}{\partial z}}_{\text{\parbox{1.75cm}{\centering Light-oriented locomotion}}}\\
\label{predator-eq}
\frac{\partial V}{\partial t} &=\underbrace{e(C_V,V)f(U)V}_{\text{\parbox{2cm}{\centering gain from  \\[-4pt] growth}}} -\underbrace{d_2(C_V,V)V}_{\text{\parbox{2cm}{\centering loss from  \\[-4pt] death}}}+\underbrace{\bigtriangledown(D_V(C_V,V)\bigtriangledown V)}_{\text{\parbox{2cm}{\centering Motility}}}\\ \nonumber &\hspace{5cm}-\underbrace{\bigtriangledown(\chi_U \frac{V}{U} \bigtriangledown U)}_{\text{\parbox{2cm}{\centering chemotaxis  \\[-4pt] towards prey}}}+\underbrace{\bigtriangledown\left(\chi_c V \bigtriangledown C_D\right)}_{\text{\parbox{2cm}{\centering chemotaxis  \\[-4pt] away from tox}}}  \\
\label{prey-tox-eq}
\frac{\partial C_U}{\partial t}&=\underbrace{a_1C_DU}_{\text{\parbox{2cm}{\centering toxicant\\[-4pt] uptake}}} -\underbrace{\sigma_1C_U}_{\text{\parbox{2cm}{\centering efflux  \\[-4pt] }}}-\underbrace{d_1(C_U,U)C_U}_{\text{\parbox{2cm}{\centering loss due to   \\[-4pt] death}}} -
\underbrace{f(U)V\frac{C_U}{U}}_{\text{\parbox{2cm}{\centering loss from  \\[-4pt] predation}}} \\
\label{pred-tox-eq}
\frac{\partial C_V}{\partial t}&=\underbrace{a_2C_DV}_{\text{\parbox{2cm}{\centering toxicant\\[-4pt] uptake}}} -\underbrace{\sigma_2C_V}_{\text{\parbox{2cm}{\centering efflux  \\[-4pt] }}}-\underbrace{d_2(C_V,V)C_V}_{\text{\parbox{2cm}{\centering loss due to   \\[-4pt] death}}} +\underbrace{f(U)V\frac{C_U}{U}}_{\text{\parbox{2cm}{\centering gain from  \\[-4pt] predation}}}\\
\label{Dissolv-tox-eq}
\frac{\partial C_D}{\partial t}&=\underbrace{\bigtriangledown(D_C\bigtriangledown C_D)}_{\text{\parbox{2cm}{\centering Diffusion}}}-\underbrace{\left(a_1U+a_2V\right) C_D}_{\text{\parbox{2cm}{\centering uptake from \\[-4pt] prey \& predator}}}+
\underbrace{\sigma_1C_U+\sigma_2C_V}_{\text{\parbox{2cm}{\centering efflux from \\[-4pt] prey \& predator}}}\\ \nonumber &\hspace{5cm}+\underbrace{d_1(C_U,U)C_U}_{\text{\parbox{2cm}{\centering prey \\[-4pt] death}}}+\underbrace{d_2(C_V,V)C_V}_{\text{\parbox{2cm}{\centering predator  \\[-4pt] death}}} 
 \end{align}
\end{subequations}

\subsection{Model simplification}
To facilitate model analysis for this manuscript, we now make a number of simplifying assumptions on the full model above. Let $C$ denote the total amount of toxicant in the system, which is divided into three pools: toxicant in the prey, toxicant in the predator, and dissolved toxicant in the medium, and therefore can be written as $C=C_U+C_V+C_D$.  While the division of toxicant across these pools depends on uptake and efflux rates, as well as population growth and death dynamics, we initially assume these are fixed proportions with the following major simplifying assumption. 
\begin{assumption}\label{assum_fixedC}
The total amount of toxicant in the system $C$, is divided between three pools with fixed proportions: toxicant inside the prey population ($\phi_1C$), toxicant inside the predator population ($\phi_2C$), and dissolved toxicant in the medium ($\phi_3C$). Where $1=\phi_1+\phi_2+\phi_3$. 
\end{assumption}
 The above assumption allows us to reduce the full model \eqref{fullmodel} down to three dimensions. Next we make simplifying assumptions on the expressions for growth and death, as well as chemo-tactic and diffusions coefficients. 

\begin{assumption}\label{assum_bu} 
The prey growth rate is not affected by the toxicant, $\alpha_2=0$. Crowding affects prey growth such that $\alpha_3=1$. We denote the maximal prey growth rate $\alpha_1$ as $\alpha$ so the prey growth expression in Eq. \eqref{newb_u} can be simplified to $$b(C_U,U, h(z,t))=\frac{\alpha h(z,t)}{1+U}.$$ 
\end{assumption}

\begin{assumption}\label{assum_ev}
 The predator functional response is scaled with half saturation constant $a=1$ and maximum ingestion rate $c$ denoted as $c_0$ such that Eq. \eqref{functionalresponse} is simplified to $f(U)=\frac{c_0U}{1+U}.$
The predator growth rate is not affected by the toxicant, $\beta_2=0$, and the conversion efficiency $\hat{e}$ is scaled to incorporate $c_0$ such that the expression for predator growth rate can be written as 
$\hat{e}\frac{U}{1+U}$.
\end{assumption}

\begin{assumption}\label{assum_death}
Positive coefficients and powers in the the toxicant dependent death rates are set as $h_1=h_2=I_1=I_2=1$, so the death rates in Eq. \eqref{newmortality} can be expressed as
\begin{align*}
d_1(C_U,U, h(z,t))&=\frac{C}{U}+m_1(h(z,t)) \\
d_2(C_V,V, h(z,t))&=\frac{C}{V}+m_2(h(z,t))
\end{align*}
where $m_1$ and $m_2$ are piecewise expressions defined in Eq. \eqref{newmortality_piecewise}. 
\end{assumption}

\begin{assumption}\label{assum_diff}
Diffusions coefficients are constants denoted as $D_1(C,U)=D_1$, $D_2(C,V)=D_2$, and $D_3(C)=D_3.$ The chemotactic coefficient for dissolved toxicant is $\chi_c=0$. 
\end{assumption}

The above simplifying assumptions allows us to reduce the full model \eqref{fullmodel} to the following three-dimensional non-linear model 
\begin{align}\label{prey-eq-1-mod-F1}
\frac{\partial U}{\partial t}&-\big[\nabla\left(D_1\cdot\nabla U\right) -v_0\partial_z U  \big] \\ 
&=\frac{h(z,t)}{1+\beta h(z,t)}\frac{U}{1+U}\left[ (\alpha-\widehat{m}_1)+\alpha\beta h(z,t)-\widehat{m}_1U\right]-\mathcal{M}_1(t)U\\ \nonumber &\hspace{2cm}-\frac{C}{U}U-c_0\frac{V}{1+U}U \nonumber \\
&=F_1(U,V,C)\cdot U \nonumber\\
\label{pred-eq-1-mod-F2}
\frac{\partial V}{\partial t}&-\left[\nabla\left(D_2\cdot\nabla V\right) - \nabla \left( \chi\frac{V}{U}\cdot \nabla U\right)\right] \\
&=\left[ \hat{e}\frac{U}{1+U}-\widehat{m}_2\frac{h(z,t)}{1+\beta h(z,t)}\right]V -\mathcal{M}_2(t)V -\frac{C}{V}V\nonumber\\
&=F_2(U,V,C)V\nonumber\\
\label{Desoled-tox-mod-F3}
\frac{\partial C}{\partial t}&-\nabla\left(D_C\cdot\nabla C\right)\\
&=\bigg[-\left(\left(a_1-\phi_1\right)U+\left(a_2-\phi_2\right)V\right)+ \frac{C}{U}+ \frac{C}{V}+[\widehat{m}_1+\widehat{m}_2]\cdot\frac{h(z,t)}{1+\beta h(z,t)} \\ \nonumber &\hspace{2cm}+\mathcal{M}_1(t)+\mathcal{M}_2(t) -a_0\bigg]\cdot C\nonumber\\
&=F_3(U,V,C)C\nonumber
\end{align}

For linearization, it is convenient to explicitly derive all the non-linear terms in the above system as the following:
\begin{align}\label{F1}
F_1(U,V,C)\cdot U &=\bigg[\frac{h(z,t)}{1+\beta h(z,t)}\frac{1}{1+U}\left[ (\alpha-\widehat{m}_1)+\alpha\beta h(z,t)-\widehat{m}_1U\right] \\ \nonumber &\hspace{3cm}-\mathcal{M}_1-\frac{C}{U}-c_0\frac{V}{1+U}\bigg ]\cdot U \nonumber \\
\label{F2}
F_2(U,V,C)\cdot V &=\left[\left[ \hat{e}\frac{U}{1+U}-\widehat{m}_2\frac{h(z,t)}{1+\beta h(z,t)}\right] -\mathcal{M}_2 -\frac{C}{V}\right]\cdot V\\
\label{F3}
F_3(U,V,C)\cdot C&=\bigg[-\left(a_1U+a_2V\right)+ \frac{C}{U}+ \frac{C}{V}+[\widehat{m}_1+\widehat{m}_2]\cdot\frac{h(z,t)}{1+\beta h(z,t)}\\ \nonumber &\hspace{3cm}+\mathcal{M}_1+ \mathcal{M}_2 -a_0\bigg]\cdot C\nonumber\\
\label{Chi}
X(U,V,\nabla U)&=\frac{V}{U}\cdot \nabla U
\end{align}
The model can then be written in terms of these non-linear functions:
\begin{align}\label{prey-eq-1-mod-F11}
&\frac{\partial U}{\partial t} -\big[\nabla\left(D_1\cdot\nabla U\right) -v_0\partial_z U  \big]=F_1(U,V,C)\cdot U\\
\label{pred-eq-1-mod-F22}
&\frac{\partial V}{\partial t}- \left[\nabla\left(D_2\cdot\nabla V\right) - \chi\nabla \cdot \left( X(U,V,\nabla U)\right)\right]=F_2(U,V,C)V\\
\label{Desoled-tox-mod-F33}
&\frac{\partial C}{\partial t}-\nabla\left(D_3\cdot\nabla C\right) =F_3(U,V,C)\cdot C
\end{align}

To this end we assume that
\begin{enumerate}\label{clasical-solution}
\item

Functions  $U(x,t), \ V(x,t), \ C(x,t)$ are twice differentiate in spacial coordinates in domain $\Omega$ for all $t>0$ and  has one derivative on the union of open time intervals $(j,j+1),$ $j\ge 0$   :

$$C_{\vec{x},t}^{2,1}\left(\Omega\times\cup_{j=1} (j,j+1)\right)$$ 

and satisfy system   \eqref{prey-eq-1-mod-F1},-\eqref{Desoled-tox-mod-F3} point wise for all $t\in \cup_{j=1} (j,j+1).$ 
\item 

In addition we will assume that all functions are in 
$C_{x,t}^{1,0}(\overline{\Omega})\times [0,\infty)$ 

\end{enumerate}

\subsection{Boundary Conditions and main assumptions} 
Let the solution of the system \eqref{prey-eq-1-mod-F1}-\eqref{Desoled-tox-mod-F3} be denoted as $\vec{U}(x,t)= (U,V,C)^T$.
On the side  $\Gamma_1$ of the boundary of the domain $\Omega$, we assume zero influx conditions for all dependent variables for all time $t>0$, 
\begin{equation}\label{BC-neuman}
\frac{\partial\vec{U}}{\partial \nu}=0. 
\end{equation} 

On the bottom and the top boundaries,  $\Gamma^H$ and $\Gamma^0$,  we assume Dirichlet conditions  (prescribed distribution of dependent variables $U,V, \ \text{and} \ C,$
\begin{align}
&\vec{U^H}(x,t)= (U^H(x,t),V^H(x,t),C^H(x,t)) \text{ -- on the bottom} \\
&\vec{U^0}(x,t)= (U^0(x,t),V^0(x,t),C^0(x,t)) \text{ -- on the top.} 
\end{align}\label{BC-dirichlet}
 
Here $U^{H,0}, \ V^{H,0}, \text{and} \ C^{H,0}$  are given functions which match the distribution of the periodic equilibrium state on the boundary $\Gamma^{H,0}$ , and $\nu$ is the outward normal to the boundary of the domain $\Omega_0$.  Let  $\vec{U}(x,t)$ be a classical  solution of the system \eqref{prey-eq-1-mod-F1},-\eqref{Desoled-tox-mod-F3}, with boundary conditions  \eqref{BC-neuman}-\eqref{BC-dirichlet} that satisfies initial condition $\vec{U}(x,z,0)=\vec{\Phi}(x).$ Again, here we denote $x=(x_1,x_2)$ and the height coordinate $x_3$ is set to be $z$.

\section{Equilibrium solutions}
\label{sec:equilibrium}

In general we define the  \textbf{equilibrium state $\vec{u}_e=\vec{u}_e(z,t)$} (``base line solution") to be a depth depending function of  $z$, and time  $t$. Below, we highlight several important features of this solution.
Let $\alpha_1=\alpha,$ $\alpha_3=1$ and $C=0$, then the base line solution for algae and daphnia depending on depth and time satisfies the following system of two equations with Dirichlet condition:

\begin{align}\label{prey-eq-1-mod}
\frac{\partial u_e}{\partial t}& -\big[\partial_z\left(D_1\cdot\partial_z u_e\right) -v_0\partial_z u_e -\mathcal{M}_1(t)u_e  \big]\\
&=\frac{h(z,t)}{1+\beta h(z,t)}\frac{u_e}{1+u_e}\big[ (\alpha-\widehat{m}_1)+\alpha\beta h(z,t)-\widehat{m}_1u_e\big]-c_0\frac{u_e}{1+u_e}v_e \nonumber, \\
\label{pred-eq-1-mod}
\frac{\partial v_e}{\partial t} &- \left[\partial_z\left(D_2\cdot\partial_z v_e\right) - \partial_z \left( \chi\frac{v_e}{u_e}\cdot \partial_z u_e\right)-\mathcal{M}_2(t)v_e \right] \\
&=\hat{e}\frac{u_e}{1+u_e}v_e-\widehat{m}_2\frac{h(z,t)}{1+\beta h(z,t)}v_e(z,t).\nonumber
\end{align}
Let the solution of the above system satisfy Dirichlet conditions:

\begin{enumerate}
\item 
\begin{equation} \label{dir_up}
u_0(0,t)= k_0e^{\omega t},\ \text{and} \ u_0(-H,t)=k_0(H)e^{\omega t}\exp(- \gamma H)
\end{equation}
\item 
\begin{equation}\label{dir_down} 
v_0(0,t)= k_1e^{\omega t}, \ \text{and} \ v_0(-H,t)=k_1(H)e^{\omega t}\exp(-\gamma H)
\end{equation}
\end{enumerate}

Here, our strategy to analyze the equilibrium state is as follows: 
1.) In section \ref{uniqness} we prove the uniqueness and stability of a positive and bounded solution of system \eqref{prey-eq-1-mod}-\eqref{pred-eq-1-mod}, where we assume a constraint on one of the solutions variables (described in Theorem \ref{uniq-eqv}). Then, in section \ref{sec:analytical} we construct analytical functions of the equilibria state assuming specific value for constants $k_1$ and $k_2$, that depend on the coefficients of the equation. In section \ref{sec:stability} we investigate the Turing stability of the constrained solution with respect to perturbations of the initial conditions. Finally, in section \ref{sims} we show that numerical solutions of the system with a similar pattern exists under much more general boundary and initial conditions.

\subsection{Uniqness of the positive and bounded solution}\label{uniqness}
Most of the theoretical results are performed on a single time-interval within a diel cycle, without switching the parameters values as defined in the piecewise function for periodic light levels, Eq. \eqref{h-def}.  For notational convenience we denote the switch time generically as $T$ and the period to be $\overline{T}$ and then consider the intervals $[0, T]$ to $[T, \overline{T}]$ separately.  In the equations, switching from ``day to night" occur only in terms $\mathcal{M}_i(t)$ as periodic step functions with time dependent switching and variations modeled by 
 $\pm\omega$, see equation \eqref{m_1^0-def}. The constructed solution is globally continuous and has all the necessary derivatives in time and space on the open time intervals $((i-1)T,iT),\ i\in \mathbb{N}$ and is unique for all time. The constructed analytical solution is increasing then decreasing on the consequent time intervals $((i-1)T,iT]$,$(iT,(i+1)T]$ for $i=1,2,3,..$. 

Throughout our analyses of the equilibrium state, we consider different cases for time and space dependence, sections \ref{Case 1} - \ref{Case 4}.  Naturally, the baseline solutions of the equilibrium state do not depend on initial data and are sliding along the time axis w.r.t. light availability. We define $\vec{u}$ as positive if each component of the vector is positive.  First we prove the maximum principle for equation \eqref{prey-eq-1-mod}. 

\begin{lem}\label{max-prin}
Assume $u(z,t)$ is a positive solution for equation \eqref{prey-eq-1-mod} in the domain $D=[-H,0]\times[0,T]$, and $l\geq v_e(z,t) \ge 0.$ Assume that $u(x,t)\geq q >0$ on the $\Gamma (D)$-parabolic boundary of the domain $D$. Then

\begin{enumerate}
\item 
There exists a constant $q_0$ depending on $T$, $l$, and $q$ and coefficients of the equation \eqref{prey-eq-1-mod} s.t. for all $(z,t)\in [-H,0]\times[0,T]$:
\begin{equation}\label{u>q0}
u(z,t)\geq e^{-\delta_0 t}q_0.
\end{equation}
Here $q_0=\min_{\Gamma(D)} u(x,t)>0,$ and $\delta_0=\widehat{m}_1\frac{\max h(z,t)}{1+\beta \min h(z,t)}+\mathcal{M}_1+cl$. 
\item Assume that in the time space domain $[-H,0]\times[0,T]$ we have 
\begin{equation}\label{delta1}
\delta_1= \alpha-\left(\mathcal{M}_1+\widehat{m}_1 \frac{h(z,t)}{1+\beta h(z,t)}\right)\ge 0,
\end{equation} 
then 
\begin{equation}\label{u<q1}
u(z,t) \le q_1e^{\delta_1 t}.
\end{equation}
Here $q_1=\max_{\Gamma(D)} u(x,t)\le C<\infty.$
\end{enumerate}
\end{lem}

\begin{proof}
Indeed let $\tilde{u}(z,t)=e^{\delta t}u(z,t).$ Then
\begin{align}
\nonumber \frac{\partial \tilde {u}}{\partial t}-\partial_z\left(D_1\cdot\partial_z \tilde{u}\right) -v_0\partial_z \tilde{u}] =e^{\delta t}\big[\frac{\partial u}{\partial t}-\partial_z\left(D_1\cdot\partial_z u\right) -v_0\partial_z u+\delta u\big]\\
=e^{\delta t}u\left[\frac{\alpha h(z,t)}{1+u}-\widehat{m}_1\frac{h(z,t)}{1+\beta h(z,t)}-\mathcal{M}_1 -c_0\frac{1}{1+u}v+\delta \right].\label{[1]}
\end{align}
Due to the positivity of function $u(z,t)$ and the condition that  $l\geq v(z,t) \geq 0$, the term in the last bracket above satisfies
\begin{equation}
\big[\cdot\big]\geq \delta -\widehat{m}_1\frac{\max h(z,t)}{1+\beta \min h(z,t)}-\mathcal{M}_1-c_l. 
\end{equation}
Consequently, function $\tilde{u}(z,t)$ is super-parabolic in $[-H,0]\times[0,T]$ if 
\begin{equation}\label{delta-0}
\delta \geq \delta_0=\widehat{m}_1\frac{\max h(z,t)}{1+\beta \min h(z,t)}+\mathcal{M}_1+c_l.
\end{equation}
Therefore, due to the maximum principle,  $\tilde{u}(z,t)\geq \min_{(z,t)\in \gamma(D)} \tilde{u}(z,t). $ From here, an estimate from below the inequality in equation \eqref{u>q0} for $u(z,t)$, with some constant $q_0$ follows. 
Similarly  taking $\delta=-\delta_1$ in the bracketed term $[\cdot]$ in equation \eqref{[1]}  with $\delta_1$ defined in equation \eqref{delta1} we can obtain an estimate $q_1$ from above in equation \eqref{u<q1}. 
\end{proof}

Next we prove the uniqueness of a given time-dependent baseline solution with some properties, under some constraint on the coefficients of system \eqref{prey-eq-1-mod}-\eqref{pred-eq-1-mod}. Under the assumption taken in Lemma \eqref{delta-0} for  $0<t\leq T$ we have
\begin{equation}\label{v/u}
\frac{v_e}{u}\leq \frac{e^{\delta_0 T}l}{q_0}=Q_q. 
\end{equation}

\begin{theorem}\label{uniq-eqv}

Let $\vec{U}_2(z,t)=(u_2(z,t),v_2(z,t))$ be a given (known) positive vector function which solves system \eqref{prey-eq-1-mod}-\eqref{pred-eq-1-mod} for  the Dirichlet boundary conditions (\ref{dir_up}, \ref{dir_down}), s.t. $\Big|\frac{\partial_z u_2}{u_2}\Big|\le \gamma$. Additionally, assume that for any solution of the baseline system  $\vec{U}_1(z,t)=(u_1(z,t),v_1(z,t))$, its second component satisfies  $v_1(z, t)\leq l$ for any $(z,t)\in D.$  Let $\vec{U}_1(z,t)$ be a classical positive solution of system \eqref{prey-eq-1-mod}-\eqref{pred-eq-1-mod} with the same Dirchlet bounday conditions and the same initial data at $t=0$ as the given positive vector function $\vec{U}_2(z,t)$. Assume the following conditions:
 
\begin{enumerate} 
\item$D_1>v_0/2,$
\item$\widehat{m}_1\frac{h(z,t)}{1+\beta h(z,t)}+\mathcal{M}_1 \ge \alpha h(z,t)-C_p(D_1-v_0/2)+v_0/2,$
\item$D_2+D_1>\delta_0\chi\gamma/2$  for some  $\delta_0>0$
\item $\widehat{m}_2\frac{h(z,t)}{1+\beta h(z,t)}+\mathcal{M}_2 \ge \hat{e}-C_p(D_2-\gamma\chi/2)+\gamma\chi/2.$
\end{enumerate}
Then $\vec{U}_2(z,t)=\vec{U}_1(z,t)$ for all time $0<t\le T$.
\end{theorem}
\begin{proof}
Let $W=u_1-u_2$,  $Z=v_1-v_2$,  and  $0<t<T$. Then, after some algebraic calculations, equation \eqref{prey-eq-1-mod} can be written as:

\begin{align}
&\frac{\partial W}{\partial t} = \frac{\alpha h(z,t)W}{(1+u_1)(1+u_2)}-\widehat{m}_1\frac{h(z,t)}{1+\beta h(z,t)}W(z,t)\nonumber \\
&-\mathcal{M}_1W -c_0\left[\frac{u_1}{1+u_1}v_1-\frac{u_2}{1+u_2}v_2\right] +\partial_z\left(D_1\cdot\partial_z W\right) -v_0\partial_z W.
\label{prey-eq-w-0}
\end{align}

It is convenient to rewrite the bracketed term above as 
\begin{equation}\label{ident-1}
\left[\frac{u_1}{1+u_1}v_1-\frac{u_2}{1+u_2}v_2\right]=\frac{u_1}{1+u_1}Z+\frac{v_2}{(1+u_1)(1+u_2)}W.
\end{equation}
Then \eqref{prey-eq-w-0}  can be written as
\begin{align}
&\frac{\partial W}{\partial t} = \frac{\alpha h}{(1+u_1)(1+u_2)}W-\big[\widehat{m}_1\frac{h}{1+\beta h}+\mathcal{M}_1\big]W -\nonumber \\
&c_0\left[\frac{u_1}{1+u_1}Z+\frac{v_2}{(1+u_1)(1+u_2)}W\right] +\partial_z\left(D_1\cdot\partial_z W\right) -v_0\partial_z W.
\label{prey-eq-w-1}
\end{align}

Let 
\begin{equation}\label{M-1-0}
M_1^0=\big[\widehat{m}_1\frac{h}{1+\beta h}+\mathcal{M}_1\big],
\end{equation}
and
\begin{equation}\label{alpha-1-0}
\alpha_1^0=\frac{\alpha h}{(1+u_1)(1+u_2)},
\end{equation}
where $h=h(z,t)$.
Because of the boundary conditions ($W|_{z=0, z=-H}=0$), after multiplying equation \eqref{prey-eq-w-1} by $w$ and integrating by parts we get
\begin{align}\label{int-w-1}
\frac12\partial_t \left(\int W^2\right)+D_1\int \left(W_z\right)^2+v_0\frac12 
\int\left(W^2\right)_z= \\
\int\alpha_1^0 W^2-\int M_1^0 W^2-c\int\frac{u_1}{1+u_2}ZW-\int\frac{v_2}{(1+u_1)(1+u_2)}W^2\nonumber
\end{align} 

Due to the boundary conditions, It is important to note that we only need to consider the integrals with respect to $z$ with bounds between $-H$ and $0$. From now on we denote  $\int:=\int_{-H}^0\left(\cdot\right) dz$. Furthermore, we have $\int:=\int_{-H}^0\left(W^2\right)_zdz = 0.$ 

Next, let us do a similar transformation to the \textsl{Daphnia} equation \eqref{pred-eq-1-mod}. Namely one can easily  see that 
\begin{align}\label{pred-eq-z-1}
\frac{\partial Z}{\partial t} = \hat{e}\left(\frac{u_1}{1+u_1}v_1-\frac{u_2}{1+u_2}v_2\right)-\left(\widehat{m}_2\frac{h}{1+\beta h}+\mathcal{M}_2\right)Z \nonumber \\
+\partial_z\left(D_2\cdot\partial_z Z\right) -\chi \partial_z \left( \frac{v_1}{u_1}\cdot \partial_z u_1-\frac{v_2}{u_2}\cdot \partial_z u_2\right)
\end{align}
Using the bracketed term described in equation \eqref{ident-1} yields

\begin{align}\label{ident-chem-tox}
\left( \frac{v_1}{u_1}\cdot \partial_z u_1-\frac{v_2}{u_2}\cdot \partial_z u_2\right)=
\frac{v_1}{u_1}\partial_zW+\frac{v_1}{u_1}\frac{\partial_zu_2}{u_2}W+\frac{\partial_zu_2}{u_2}Z
\end{align}

which can be simplified to 

\begin{align}\label{pred-eq-z2}
\frac{\partial Z}{\partial t} = \hat{e}\left[\frac{u_1}{1+u_1}Z+\frac{v_2}{(1+u_1)(1+u_2)}W\right]-\left(\widehat{m}_2\frac{h}{1+\beta h}+\mathcal{M}_2\right)Z\\
+\partial_z\left(D_2\cdot\partial_z Z\right) -\chi \partial_z\left( \frac{v_1}{u_1}\partial_zW+\frac{v_1}{u_1}\frac{\partial_zu}{u_2}Z+\frac{\partial_zu}{u_2}Z\right).\nonumber
\end{align}

Multiplying this equation by $Z$, integrating by parts, and taking into account the boundary conditions yields the following:

\begin{align}\label{int-z-1}
\frac12\frac{\partial (\int Z)^2}{\partial t} &= \hat{e}\int\left[\frac{u_1}{1+u_1}Z^2+\int\frac{v_2}{(1+u_1)(1+u_2)}WZ\right]\\\nonumber
&\qquad-\int\left(\widehat{m}_2\frac{h}{1+\beta h}+\mathcal{M}_2\right)Z^2-D_2\int\left(\partial_z Z\right)^2 \\
&\qquad+\chi\left( \int \frac{v_1}{u_1}\partial_zW\partial_z Z+\int\frac{v_1}{u_1}\frac{\partial_zu}{u_2}W\partial_z Z+\int\frac{\partial_zu_2}{u_2}Z\partial_z Z\right)\nonumber
\end{align}
Adding equations \eqref{int-w-1} and \eqref{int-z-1} yields

\begin{align}\label{w+z}
&\frac12\frac{\partial (\int Z)^2}{\partial t}+\frac12\frac{\partial (\int W)^2}{\partial t}+D_1\int \left(W_z\right)^2+D_2\int\left(\partial_z Z\right)^2 \\
&-\chi\left( \int \frac{v_1}{u_1}\partial_zW\partial_z Z-\int\frac{v_1}{u_1}\frac{\partial_zu}{u_2}W\partial_z Z-\int\frac{\partial_zu_2}{u_2}Z\partial_z Z\right)\nonumber\\
&\qquad\leq\int\alpha_1^0 W^2-\int M_1^0 W^2-c\int\frac{u_1}{1+u_2}ZW-\int\frac{v_2}{(1+u_1)(1+u_2)}W^2\nonumber\\
&\qquad\qquad+\chi\left( \int \frac{v_1}{u_1}\partial_zW\partial_z Z+\int\frac{v_1}{u_1}\frac{\partial_zu}{u_2}W\partial_z Z+\int\frac{\partial_zu_2}{u_2}Z\partial_z Z\right)\nonumber.
\end{align}

Using the Cauchy inequality with $\epsilon$, our assumption on $v_1\leq l$, and equation \eqref{v/u} we rewrite \eqref{w+z} as

\begin{align}\label{w+z-1}
&\frac12\frac{\partial (\int Z)^2}{\partial t}+\frac12\frac{\partial (\int W)^2}{\partial t}+(D_1-\frac{\chi Q_0}{2})\int \left(W_z\right)^2\\
&+(D_2-\chi\frac{ Q_0+\gamma^2}{2})\int\left(\partial_z Z\right)^2 -\frac{\chi}{2}\int W^2-\frac{\chi}{2}\int Z^2
\nonumber\\
&\qquad\leq\int\left(\alpha_1^0+\frac{c}{2}-M_1^0\frac{v_2}{(1+u_1)(1+u_2)}\right)W^2+\frac{c}{2}\int Z^2.\nonumber
\end{align}
Now, applying the Poincare inequality yields
\begin{align}\label{w^2+z^2-only}
&\frac12\frac{\partial (\int Z)^2}{\partial t}+\frac12\frac{\partial (\int W)^2}{\partial t}-\left|(D_1-\frac{\chi Q_0}{2})\right|\int \left(W\right)^2\\
&-\left|(D_2-\chi\frac{ Q_0+\gamma^2}{2})\right|C_p\int\left( Z\right)^2 -\frac{\chi}{2}\int W^2-\frac{\chi}{2}\int Z^2\nonumber\\
&\qquad\leq\int\left(\alpha_1^0+\frac{c}{2}-M_1^0\frac{v_2}{(1+u_1)(1+u_2)}\right)W^2+\frac{c}{2}\int Z^2.\nonumber
\end{align}
\end{proof}
\begin{rem}
It is important to mention that if $u_2(z,t)$ satisfies the constraint $|\frac{\partial{u_2}(z,t}{u_2(z,t)}|\leq \gamma$ then there are no constraints on the relationship between parameters of the baseline model and the constant $Q_0$, which bounds the $v_1(z,t)$ component of the solution $\vec{U}_2(z,t)$.  Namely, assume that $v_1(z,t)\leq Q_0<\infty$ for all $t\in [0,T]$ then the following hold:
\begin{enumerate}
\item The uniqueness theorem on the time interval $[0,\tau]$ for the main positive baseline problem follows from equation \eqref{w^2+z^2-only}. 
 \item 
 Moreover, if the initial data for the vector functions $\vec{U}_1(z,t)$ and $\vec{U}_2(z,t)$ are not the same, then there exists a constant $B$ on the  interval $[0,T]$ that depends on $Q_0$ and constants of the baseline system s.t.
 \begin{equation}
 \int W^2(\cdot,t)+Z^2(\cdot,t) \leq e^{B t} \int W^2(\cdot,0)+Z^2(\cdot,0)
 \end{equation}
\end{enumerate}
\end{rem}

\subsection{Analytical representation of equilibrium}\label{sec:analytical}
In this section we present analytical representations of the the equilibrium state for four cases, considering the temporal and spatial dependence of the solution. 
\subsubsection{ Case 1: Time and spatial dependent equilibrium $\omega\neq 0$, $\gamma\neq 0$}\label{Case 1}
Here we show that under some constraints on the parameters, we can write a ``baseline" solution $\vec{u_0}$ of the system \eqref{prey-eq-1-mod} and   \eqref{pred-eq-1-mod},  which is time and spatial dependent. 

Let 
\begin{equation}\label{u0-def}
\vec{u}_0= \vec{u}_0(z,t)= \left\{
        \begin{array}{ll}
          (\nu_1e^{\omega t}\exp(\gamma z)\ ;\ \nu_2e^{\omega t}\exp(\gamma z)\ ; \ 0)& \quad \ 0\leq t\leq 1\\
          (\nu_1e^{\omega(2-t)}\exp(\gamma z) \ ; \  \nu_2e^{\omega(2-t)}\exp(\gamma z) \ ; \ 0) & \quad 1 \leq t \leq 2. 
        \end{array}
    \right.
\end{equation}

The goal is to find constraints on the coefficients of the system that yield a non-negative vector solution of the system. 
Substituting $\vec{u}_0$ into equations \eqref{prey-eq-1-mod} and \eqref{pred-eq-1-mod}, assuming that $0<t<1,$ and some simplification yields an equation for the prey equilibrium:
\begin{align}\label{prey_bl_3}
&0=\left(\left(\alpha-c_0\nu_2-\widehat{m}_1\right) +\left((\alpha-c_0\nu_2)\beta-\nu_1\widehat{m}_1\right) e^{\omega t}\exp(\gamma z)\right)\times\\
&\frac{e^{\omega t}\exp(\gamma z)}{\left(1+\nu_1e^{\omega t}\exp(\gamma z)\right)\left(1+\beta e^{\omega t}\exp(\gamma z)\right)}
+\left(D_1\gamma^2 -v_0\gamma - (\mathcal{M}_1+\omega)\right)\nu_1,\nonumber
\end{align}
and an equation for the predator equilibrium
\begin{align}\label{pred-1_bl_3}
 0 = \left(\frac{\hat{e}\nu_1-\widehat{m}_2+\nu_1(\hat{e}\beta-\widehat{m}_2)(e^{\omega t}\exp(\gamma z)}{(1+\nu_1e^{\omega t}\exp(\gamma z)(1+\beta e^{\omega t}\exp(\gamma z))}\right)e^{\omega t}\exp(\gamma z)\nu_2
 \\
+\left((D_2 -\chi)\gamma^2-(\mathcal{M}_2+\omega)\right)\nu_2. \nonumber
\end{align}

Let us first find $\nu_1$, $\nu_2$, $\gamma$, and $\omega$ s.t. the prey's equation  \eqref{prey_bl_3} is satisfied. It makes sense biologically, to satisfy this equation first since algae can exists without \textsl{Daphnia}. In our model, \textsl{Daphnia} is dependent on algae and the population will go to zero if the algae population goes to zero. 
Recalling that $\mathcal{M}_i+\omega =\overline{m}_i$ and since $\omega\neq 0$ and $\gamma \neq 0$  it is clear that in order to satisfy the prey equation \eqref{prey_bl_3}, it is necessary and sufficient that
\begin{align}\label{prey-1_bl_3}
 0=\alpha-c_0\nu_2-\widehat{m}_1\\
 0=(\alpha-c_0\nu_2)\beta-\nu_1\widehat{m}_1\\
 0=\left(D_1\gamma^2 -v_0\gamma - \overline{m}_1\right).
\end{align}
Let
\begin{equation}\label{gamma_eq}
\gamma=\frac{v_0+\sqrt{v_0^2+4\overline{m}_1 D_1}}{2D_1},
\end{equation}
and
\begin{equation}\label{gamma^2_eq}
\gamma^2=\frac{v_0\gamma+\overline{m}_1}{D_1},
\end{equation}
then due to equation \eqref{m_1^0-def}
\begin{align}
D_1\gamma^2-v_0\gamma-\overline{m}_1=0 
\end{align}
and consequently the following should hold
\begin{equation}\label{nu_2_cond}
\alpha-c_0\nu_2-\widehat{m}_1=0, \  \text{or} \ \nu_2=\frac{\alpha-\widehat{m}_1}{c_0},
\end{equation}
and
\begin{equation}\label{nu_1_cond}
 \nu_1=\beta.
\end{equation}
In order to satisfy the equation for \textsl{Daphnia} we note $\mathcal{M}_1+\omega =\overline{m}_2$ and impose strict conditions on some parameters of the model. Namely for equation \eqref{pred-1_bl_3} to hold, it is necessary and sufficient that 
\begin{align}\label{pred-1_bl_3a}
0 = \hat{e}\nu_1-\widehat{m}_2
\\
0=\hat{e}\beta-\widehat{m}_2
\\
0=(D_2 -\chi)\gamma^2-\overline{m}_2.
\end{align}
Consequently, taking into account the above equations for $\nu_1$, $\nu_2$, $\gamma$, and $\omega$ (equations \eqref{nu_2_cond}, \eqref{nu_1_cond}, and \eqref{gamma_eq}) we get the following algebraic relationship on the parameters for predator:
\begin{align}\label{cond_pred_eq_1}
0=\hat{e}\beta-\widehat{m}_2
\\
0=-\overline{m}_2+(D_2 -\chi )\frac{v_0\left(v_0+\sqrt{v_0^2+4\overline{m}_1D_1}\right)+\overline{m}_12D_1}{2D_1^2}
\end{align}
Similar results hold for the time interval  $1\leq t\le 2$, where the equations are similar but with  $\omega$ substituted by $-\omega$. The following theorem summarizes the results. 
\begin{theorem}\label{gamma-neq-0-omega-neq-0}
Let \begin{equation}\label{alpha>m_1}
\alpha>\widehat{m}_1.
\end{equation} 
There exists:
\begin{enumerate}
\item parameters   $\widehat{m}_1>0$ and $\widehat{m}_2>0$, which define $d_1(0,U,z)$ and $d_2(0,U,z)$
\item parameters $\gamma$ and $\omega$ from \eqref{gamma_eq},
\item and coefficients  from equations in \eqref{cond_pred_eq_1}, 
\end{enumerate}
such that  $\nu_1$ from \eqref{nu_1_cond} and $\nu_2$ from \eqref{nu_2_cond} are positive, and  function $\vec{u}_0$ in \eqref{u0-def} satisfies the system of equations \eqref{prey-eq-1-mod}-\eqref{pred-eq-1-mod} for all  $0<t \leq \overline{T}=2$. 
\end{theorem}

\begin{rem}
The baseline solution satisfies the conditions on the component $u_2(z,t)$ in Theorem \ref{uniq-eqv}. Consequently if parameters of the equations satisfy conditions in the above Theorem \ref{gamma-neq-0-omega-neq-0}, then this baseline solution also satisfies the Dirichlet problem in the cylinder $U\times (t_0,T]$ for all time and at any point of the domain of interest. 
\end{rem}

\subsubsection{Case 2: Time and spatial independent equilibrium $\omega= 0  \ \gamma= 0$} \label{Case 2}
Below we prove the existence of a homogeneous equilibrium state under very generic assumptions on the parameters of the model. Here, we also assume the dissolved toxicant concentration $C=0$. 
\begin{theorem}
Let $\omega= 0$ and $\gamma= 0. $ Assume  that 
\begin{equation}\label{eqv_hom_cond_1}
S_0=\hat{e}(1+\beta)-\widehat{m}_2>0, \ \text{and} \ \alpha-\overline{m}_1>0, 
\end{equation}
and
\begin{equation}\label{eqv_hom_cond_2}
\alpha>\frac{\widehat{m}_1}{1+\beta}+\frac{c\widehat{m}_2}{S_0},
\end{equation}

then exist a constant vector solution $\vec {u_0}=\left(\nu_1;\nu_2,0\right)$ for system \eqref{prey-eq-1-mod-F11}-\eqref{Desoled-tox-mod-F33}, with   $\nu_1>0,$ and $\nu_2>0.$ 
\end{theorem}
\begin{proof}
For $\gamma=0,$ and $\omega =0$, the parameters are not depending on time, 
and from \eqref{prey_bl_3} and \eqref{pred-1_bl_3} we end up with only two equations for two unknowns $\nu_1$ and $\nu_2$:
\begin{align}\label{prey-1_gamma=0_omega=0}
0=\alpha-c_0\nu_2-\widehat{m}_1 +(\alpha-c_0\nu_2)\beta-\nu_1\widehat{m}_1\\
 0=\hat{e}\nu_1-\widehat{m}_2+\nu_1(\hat{e}\beta-\widehat{m}_2).
\end{align}
\end{proof}
Here 
\begin{equation}\label{nu_1_gamma=0_omega=0}
\nu_1=\frac{\widehat{m}_2}{\hat{e}(1+\beta)-\widehat{m}_2}, 
\end{equation}
and substituting $\nu_1$ into \eqref{nu_1_gamma=0_omega=0} results in 
\begin{equation}
\nu_2=\frac{\alpha-\frac{\widehat{m}_1}{1+\beta}-\frac{c\widehat{m}_2}{S_0}}{c_0}. 
\end{equation}

\subsubsection{Case 3: Time dependent  and spatial independent equilibrium $\omega\neq 0,\ \gamma= 0$}\label{Case 3}
Below we prove the existence of the equilibrium  state which is time dependent but spatially homogeneous under exact assumptions on model parameters. Again, we assume the dissolved toxicant concentration $C=0$. 

\begin{theorem}
Let $\omega\neq 0$ and $\gamma= 0$. Assume that 
\begin{equation}\label{eqv_hom_cond_2b}
\alpha >\widehat{m}_1
\end{equation}
then exist a constant vector solution $\vec {u_0}=\left(\nu_1;\nu_2,0\right)$ for the system \eqref{prey-eq-1-mod-F11}-\eqref{Desoled-tox-mod-F33}, with   $\nu_1>0,$ and $\nu_2>0.$
\end{theorem}
\begin{proof}
Explicit positive expression for $\nu_1$, and $\nu_2$ immediate follow from equations \eqref{nu_2_cond} and \eqref{nu_1_cond}.
\end{proof}

\begin{rem}
In this case of only time-dependent coefficients, it is clear that for appropriate non-zero initial conditions and relationships between the constants of the coefficients $\left(\widehat{m}_1,\widehat{m}_2, \overline{m}_1, \overline{m}_2, c_0, \hat{e}\right)$  there exists a time oscillating vector $(u_e(t),v_e(t)$ solution of system \eqref{prey-eq-1-mod}-\eqref{pred-eq-1-mod} that stabilizes in time for any periodic function $h(t)=h(\cdot,t)$. This can be rigorously proved via qualitative analyses of the system. It is important to note that this is true even for non-periodic but appropriately oscillating function $h(t)$. For example, this holds true for almost periodic functions with periods bounded from below and above. We leave this analyses for future work, but illustrate the conjecture with numerical simulations in section \ref{sims}. 
\end{rem}

\subsubsection{Case 4: Time independent and spatial dependent equilibrium $\omega= 0  \ \gamma \neq 0$}\label{Case 4}
Below we prove the existence of an equilibrium state which is time dependent but spatially homogeneous  under exact assumption on the model parameters, and $C=0$. 

\begin{theorem}
Let $\omega = 0 \ \gamma\neq 0.$ 
Assume that
\begin{equation}\label{eqv_hom_cond_2c}
\alpha>\widehat{m}_1
\end{equation}
and all coefficients satisfy the conditions in Theorem \ref{gamma-neq-0-omega-neq-0}.
Then exist an equilibrium constant vector $\vec {u_0}=\left(\nu_1;\nu_2,0\right)$ for system \eqref{prey-eq-1-mod-F11}-\eqref{Desoled-tox-mod-F33}, with   $\nu_1>0,$ and $\nu_2>0.$
\end{theorem}
\begin{proof}
Explicit positive expression for $\nu_1$, and $\nu_2$ immediate follow from equations \eqref{nu_2_cond} and \eqref{nu_1_cond}.
\end{proof}

In the next section we consider case 1 (subsection \ref{Case 1}) and investigate the stability of the time and spatial dependent equilibrium state vector $\vec{u}_0(z,t)$ defined in \eqref{u0-def} for algae and \textsl{Daphnia} using the energy method technique.


\subsection{Stability of the time and spatial dependent solution}\label{sec:stability}
In this section we conduct a Tuning stability analyses of our time and spatial dependent equilibrium, with respect to our reaction diffusion-chemo-tactic  system, which for convenience we rewrite here. First, we define functions in the right hand side (RHS) as reactive terms: 
\begin{align*}
\Phi_1(U,V,C)&=F_1(U,V,C)U \\
\Phi_2(U,V,C)&=F_2(U,V,C)V\\
\Phi_3(U,V,C)&=F_3(U,V,C)C 
\end{align*}
and function $X(U,V,\vec{\xi})$  in the left hand side (LHS) of equation \eqref{pred-eq-1-mod-F22} as a chemotactic term. For convenience, we keep the linear part of the diffusion terms in the LHS of each equation. Our system can then be written as:

\begin{align}\label{prey-eq-1-mod-F11a}
&\frac{\partial U}{\partial t} -\big[\nabla\left(D_1\cdot\nabla U\right) -v_0\partial_z U  \big]=F_1(U,V,C)\cdot U\\
\label{pred-eq-1-mod-F22a}
&\frac{\partial V}{\partial t}- \left[\nabla\left(D_2\cdot\nabla V\right) - \chi_0\nabla \cdot \left( \vec{X}(U,V,\nabla U)\right)\right]=F_2(U,V,C)V\\
\label{Desoled-tox-mod-F33a}
&\frac{\partial C}{\partial t}-\nabla\left(D_C\cdot\nabla C\right) =F_3(U,V,C)\cdot C
\end{align}

Once again, we denote vector $\vec{U}(z,t)=\big(U,V,C\big)^T$ as an arbitrary solution of the system and  $\vec{u}_0(z,t)=\big(u_e,v_e,c_e\big)^T$ as the time and spatially dependent equilibrium state. Here, our Turing stability analyses is based on the Carath{\'e}odory principle, where the Carath{\'e}odory theorem for functions $\Phi_i(U,V,C)$ takes the form: 
\begin{align}\label{carath1}
\Phi_i(U,V,C)- \Phi_i(u_e,v_e,c_e)&=\psi^1_i(U,V,C)(U-u_e)+\psi^2_i(U,V,C)(V-v_e)\\
&\qquad+\psi^3_i(U,V,C)(C-c_e). \nonumber
\end{align}
In above in the Carath{\'e}odory functions, $\psi_i^k$, the upper indices $k=1,2,3$ correspond with variables $U$, $V$, and $C$ respectively, and the lower indices $i=1,2,3$ correspond with the reactive term functions $F_1$, $F_2$, and $F_3$ respectively. We denote the following limits: 
\begin{align}
& \psi^1_{i,0}(z,t)=\lim_{\vec{U}\to \vec{u}_0} \psi^1_i(\vec{U})=\left.\frac{\partial \left( F_i(\vec{U})\cdot U\right) }{\partial U}\right\vert_{\vec{U}=\vec{u}_0} \label{F_U}\\
& \psi^2_{i,0}(z,t)=\lim_{\vec{U}\to \vec{u}_0} \psi^2_i(\vec{U})=\left.\frac{\partial \left( F_i(\vec{U})\cdot V\right) }{\partial V}\right\vert_ {\vec{U}=\vec{u}_0}\label{F_V}\\
& \psi^3_{i,0}(z,t)=\lim_{\vec{U}\to \vec{u}_0} \psi^3_i(\vec{U})=\left.\frac{\partial \left( F_i(\vec{U})\cdot C\right) }{\partial C}\right\vert_ {\vec{U}=\vec{u}_0} \label{F_C}.
\end{align}
Now, applying the Carath{\'e}odory Theorem to the non-linear chemotactic vector term $\vec{X}(U,V,\nabla) U=\frac{V\nabla U}{U}$ yields
\begin{align}\label{carath2}
\vec{X}(U,V,\vec{\xi})- \vec{X}(u_e,v_e,\vec{\xi}_e)&=\vec{\chi}^1(U,V,\vec{\xi})(U-u_e)+\vec{\chi}^2(U,V,\vec{\xi})(V-v_e)\\
&\qquad+\tilde{\chi}^3(U,V,\vec{\xi})(\vec{\xi}-\vec{\xi}_e).\nonumber
\end{align}
Here functions $\vec{X}(U, V, \vec{\xi})$ and $\vec{\chi}^k(U, V, \vec{\xi})$ for $k=1, 2$ are vectors and 
$\tilde{\chi}^3(U,V,\vec{\xi})$ is a diagonal matrix. Note that $\vec{\xi}=\nabla U$. The functions have the following properties:


\begin{align}
&\vec{\chi}^1_0(z,t)= \lim_{\nabla U\to \nabla u_0 }\lim_{\vec{U}\to \vec{u}_0} \vec{\chi}^1(U,V,\nabla U) = \left.\left(-\frac{V\nabla U}{U^2} \right)\right\vert_ {\vec{U}=\vec{u}_0; \nabla U= \nabla u_e } \label{chi_U}\\
&\vec{\chi}^2_0(z,t)=\lim_{\nabla U\to \nabla u_0 }\lim_{\vec{U}\to \vec{u}_e} \vec{\chi}^2(U,V,\nabla U)= \left.\left(\frac{\nabla U}{U}\right) \right\vert_ {U=u_0; \nabla U= \nabla u_e }  \label{chi_V}\\
&\tilde{\chi}^3_0(z,t)=\lim_{\nabla U\to \nabla u_0 } \lim_{\vec{U}\to \vec{u}_e} \tilde{\chi}^3(U,V,\nabla U) =\left.\frac{U}{V}\mathbb{I}\right\vert_ {\vec{U}=\vec{u}_0} \label{chi_nabla}
\end{align}
where $\mathbb{I}$ is the identity matrix. 
Note that in above \eqref{chi_V} in the RHS does not contain $V$ and therefore $U$ and $u_e$ are just scalars. 

\begin{rem}
Note that the equilibrium state $\vec{u}_0 =(u_e(z,t),v_e(z,t),c_e=0)$ depends on two variables, time $t$ and depth $z$, in the generic light-driven model. Thus all of the parameters of  $\psi^k_{i,0}$ and $\chi^k_0$ are functions of time $t$ and depth $z$.  Using the Carath{\'e}odory theorem for the linearization process gives a linear system of equations for $\vec{u}=\big(U-u_e,V-v_e,C-c_e\big)$, whose coefficients generally depend on $t$ and $z$.
\end{rem}

For convenience, let
\begin{equation}\label{uvc} 
u=U-u_e, \,\ v=V-u_e, \, \text{and}\ c=C-c_e 
\end{equation}
then the system can be written in the following form:
\begin{align}
\label{prey_ln}
\frac{\partial u}{\partial t}& -\big[\nabla\left(D_1\cdot\nabla u\right) -v_0\partial_z u  \big]=\psi^1_{1,0}(z,t)u+\psi^2_{1,0}(z,t)v+\psi^3_{1,0}(z,t)c   \\
\label{pred_ln}
\frac{\partial v}{\partial t} &- \nabla\left(D_2\cdot\nabla v\right) + \chi\nabla\cdot \left( \vec{\chi}^1_0(z,t)u+\vec{\chi}^2_0(z,t)v+\tilde{\chi}^3_0(z,t)\cdot(\nabla u)\right) \\
&\qquad=\psi^1_{2,0}(z,t)u+\psi^2_{2,0}(z,t)v+\psi^3_{2,0}(z,t)c.\nonumber\\
\label{toxin_ln}
\frac{\partial c}{\partial t} &-\nabla\left(D_C\cdot\nabla c\right)=\psi^1_{3,0}(z,t)u+\psi^2_{3,0}(z,t)v+\psi^3_{3,0}(z,t)c. 
\end{align}

We assume that perturbations from the equilibrium state only occur at the initial conditions, therefore $u(x,0)\neq 0$,  $v(x,0)\neq 0$, and $c(x,0)\neq 0$. 
Additionally $c(x,0)>0$ whether the perturbations in the prey ($u$) and in predator ($v$) are positive or negative. We assume that the population densities and the concentration of toxins on the top and bottom boundaries match those of the equilibrium state. Therefore perturbations of $u$, $v$, and $c$ are assumed to be zero on the top and bottom boundaries. 
Additionally, we assume that the side boundary of the cylindrical domain serves as an insulator and take the fluxes on the side boundary to be zero. 

These assumptions on the initial and boundary conditions of the linearized system \eqref{prey_ln}-\eqref{toxin_ln} are summarized below:
\begin{align}
\label{IC-lin}
&u(x,0)=u^0(x), \ v(x,0)=v^0(x) ,\  c(x,0)=c^0(x) \\ 
\label{Neum-cond-lin}
&\frac{\partial u }{\partial \vec{n}}= \frac{\partial v }{\partial \vec{n}}=\frac{\partial c }{\partial \vec{n}}=0  \ \text {on side boundary } \Gamma_1\times (0,\infty) \\
\label{top-cond}
&u(x_1,x_2, 0,t)=  v(x_1,x_2, 0,t)=c(x_1,x_2, 0,t)=0 \ \text{on} \ \Gamma^0\times (0,\infty)  \\
\label{tbotom-cond}
& u(x_1,x_2, -H,t)=v(x_1,x_2, -H,t)= c(x_1,x_2, -H,t)=0 \ \text{on} \ \Gamma^H\times (0,\infty)
\end{align}
The following proposition is a straight forward application of Green's formula.   
\begin{prop}
Green's formula yields the following integral identity for the components of vector $\vec{u}$ and the gradients of these components: 
\begin{align}\label{main-identity} 
&\frac{d}{dt}\left(\int_{\Omega} (u^2(x,t)+v^2(x,t)+ c^2(x,t)) dx\right)+\\ \nonumber
&\int_{\Omega}\left(D_1 (\nabla u)^2 +D_2\nabla v)^2  + D_3 (\nabla c)^2 \right) dx=\\ \nonumber 
&\int_{\Omega} \psi^1_{1,0}(x,t) u^2+\psi^2_{1,0}(x,t) uv+\psi^3_{1,0}(x,t) cu dx+\\ \nonumber
&\int_{\Omega}\psi^1_{2,0}(x,t) uv+\psi^2_{2,0}(x,t) v^2+\psi^3_{2,0}(x,t) cv dx+\\ \nonumber
&\int_{\Omega}\psi^1_{3,0}(z,t) uc+\psi^2_{3,0}(x,t) vc+\psi^3_{3,0}(x,t) c^2dx+\\ \nonumber
&\int_{\Omega}\chi\left[\left(\vec{\chi}^1_0(z,t)\cdot \nabla v\right) u+\left(\vec{\chi}^2_0(z,t)\cdot \nabla v\right) v+\left( \tilde{\chi}^3_0(z,t) \nabla u\right)\cdot \nabla v\right]dx. \nonumber
\end{align}
\end{prop}
\begin{proof}
To obtain the integral identity \eqref{main-identity} we multiply the first equation by $u(x,t)$, the second equation  by $v(x,t)$, and the third equation by $c(x,t)$ and then apply Green's formula. All boundary term vanish due to the boundary conditions \eqref{Neum-cond-lin}-\eqref{tbotom-cond}.
\end{proof}

In order to investigate the stability of the equilibrium state, we explore whether the system will return to the equilibrium state after a given perturbation. Assume that a perturbation of the equilibrium state occurs at the initial time, $t=0$. The stability properties, or whether the system will return to the equilibrium state can depend on the geometry of the aquifer, as well as the motility of the algae and \textsl{Daphnia}. The question we focus on here is: Can the geometrical depth of the aquifer and/or the migration capabilities of the populations, with respect to the toxic coefficient and reactive terms, be such that the system will be restored to the equilibrium state? 

Mathematically speaking for any $\bold{``not \ too \ big"}$  perturbation of the system, can we determine large enough motility coefficients and/or a small enough reservoir depth such that the system stabilized in time to $\vec{u}_0$? In the Turing stability analyses,  $\bold{``not \ too  \ big"}$ means that the perturbation is such that oscillations of the nonlinear functions $F^s$ and $X^s$ near equilibrium are limited so that the functions $F^s(\vec{U})$ and $X^s(\vec{U})$  can be effectively quantified by values of these functions at equilibrium  $\vec{u}_0 .$

The answer to this question is positive in the following sense: for any value of the chemotactic functions $X^s$ and reactive functions $F^s$ at equilibrium state $\vec{u}_0$, there exists diffusion coefficients $D_1$, $D_2$, and $D_3$ such that the perturbed solution will stabilize to the equilibrium state, as long as the perturbation is $\bold{``not \ too \ big"}$. These ideas are summarized in the following theorems. 

\begin{theorem}\label{estimate-wrt-in-data}
Let
\begin{align} 
\psi^1_1&=\max |\psi^1_{1,0}(z,t)| \ ; \ \psi^1_2=\max |\psi^1_{2,0}(z,t)| \ ; \  \psi^1_3=\max |\psi^1_{3,0}(z,t)| \label{psi^1},  \\
\psi^2_1&=\max |\psi^2_{1,0}(z,t)| \ ; \ \psi^2_2=\max |\psi^2_{2,0}(z,t)| \ ; \ \psi^2_3=\max |\psi^2_{3,0}(z,t)| \label{psi^2},\\
\psi^3_1&=\max |\psi^3_{1,0}(z,t)| \ ; \ \psi^3_2=\max |\psi^3_{2,0}(z,t)|  \ ; \  \psi^3_3=\max |\psi^3_{3,0}(z,t)|\label{psi^3},
\end{align}
and
\begin{equation}\label{chi}
\chi^1=\chi \max |\vec{\chi}^1_0| \ ; \ \chi^2=\chi \max |\vec{\chi}^2_0|\ ; \ \chi^3=\chi \max ||\tilde{\chi}^3_0||.
\end{equation}
In above equations the $\max$ is taken w.r.t $z\in [-H,0]$ and $t\in [0,\infty)$. Due to properties of the equilibrium state $\vec{u}_0$, these maximums are positive and bounded constants.
Then due to the linearity of the solution to the linearized system  \eqref{prey_ln}-\eqref{toxin_ln}, the following inequality holds
\begin{align}\label{main_int_ineq}
&\frac{d}{dt}\int_{\Omega} \left(u^2+v^2+ c^2\right) dx\\
&\qquad+\int_{\Omega}\left[\left(\frac{D_1}{2}-\frac{\chi^3}{2}\right) \cdot (\nabla u)^2 +\left(\frac{D_2}{2}-\frac{\chi^1+\chi^2+\chi^3}{2}\right) \cdot (\nabla v)^2\right]dx \nonumber\\
&\qquad+\int_{\Omega}\left[\left(\frac{D_2}{2}-\frac{H^2\chi^2}{2}\right) \cdot (\nabla v)^2+\left(\frac{D_1}{2}-\frac{H^2\chi^1}{2}\right) \cdot (\nabla u)^2 \right]+\frac{D_3}{H^2}  c^2dx\nonumber\\ 
&\leq\int_{\Omega}\left( \psi^1_1+\frac{\psi^2_1+\psi^3_1+\psi^1_2+\psi^3_2 +\psi^1_3}{2}\right)u^2dx\nonumber\\
&\qquad+\int_{\Omega}\left(\psi^2_2+\frac{\psi_1^2+\psi_2^1+\psi_3^2+\psi^2_3}{2} \right)v^2 dx\nonumber\\
&\qquad+\int_{\Omega}\left(\psi^3_3+\frac{\psi_3^1+\psi_3^2+\psi^3_1+\psi^2_3}{2} \right)c^2 dx.\nonumber
\end{align}
 
\end{theorem} 
In above  $H^2$ is a constant in the Poincare Inequality: $\int_{\Omega} u^2 dx\leq H^2 \int_{\Omega} |\nabla u|^2 dx$ for the class of $W_2^1$ which are vanishing on top and the bottom of the domain (aquatic reservoir) of the depth $H$. The fraction $\frac{1}{2}$ comes from an application of the standard Cauchy inequality.  

To study stability, for compactness, we denote the terms in RHS as  
\begin{align}\label{psi-chi-u}
&q_u=\left( \psi^1_1+\frac{\psi^2_1+\psi^3_1+\psi^1_2+\psi^3_2 +\psi^1_3}{2}\right),\\
&q_v=\left(\psi^2_2+\frac{\psi_1^2+\psi_2^1+\psi_3^2+\psi^2_3}{2} \right),\label{psi-chi-v}\\
&q_c=\left(\psi^3_3+\frac{\psi_3^1+\psi_3^2+\psi^3_1+\psi^2_3}{2} \right). \label{psi-chi-c}
\end{align}
Using the above differential inequality \eqref{main_int_ineq}, one can easily get stability in finite time w.r.t. perturbation in the initial data. 

\begin{theorem}\label{estimates-I-alpha}
Let $\vec{u}$ be a solution of the linearized system \eqref{prey_ln}-\eqref{toxin_ln}. Assume that 
the coefficients of the equations and depth of the reservoir $H$ are such that
\begin{align}\label{Du-chi-cond}
&Q_u=\left(D_1-\frac{\chi^3}{2}-\frac{H^2\chi^1}{2}\right)>0,\\
&Q_v=\left(D_2-\frac{H^2\chi^2}{2}-\frac{\chi^1+\chi^2+\chi^3}{2}\right)>0,\label{Dv-chi-cond}\\
&Q_c=\frac{D_3}{H^2}. \label{Dc-cond}
\end{align}
Denote $I(t)=\int_{\Omega} \left(u^2+v^2+ c^2\right) dx.$
Then 
\begin{equation}\label{alpha-estimate}
I(t)\leq I(0)\cdot e^{\alpha t}.
\end{equation}
Here $\alpha$ is a constant depending on the constants $q_s$ and $Q_s$, which can be positive or negative.
\end{theorem}

If diffusion and the size of the domain weakly dominate the chemotactic and reactive forces, then the $L_2$ Lyapunov functional of the deviation between the perturbed solution and the equilibrium state solution, $I(t)$,  does not exceed the initial perturbation.

\begin{theorem}\label{Stability}
Let $\vec{u}$ be a solution of the linearized system \eqref{prey_ln}-\eqref{toxin_ln}. Assume that
the coefficients of the equations and depth of the reservoir $H$ are such that
\begin{align}\label{Du-chi-conda}
&Q_u\geq q_u,\\
&Q_v\geq q_v,\label{Dv-chi-conda}\\
&Q_c\geq q_c. \label{Dc-conda}
\end{align}
Then 
\begin{equation}\label{alpha-estimatea}
I(t)\leq I(0).
\end{equation}

\end{theorem}

If diffusion and the size of the domain strongly dominate the chemotactic and reactive forces, then the $L_2$ Lyapunov functional of the deviation between the perturbed solution and the equilibrium state solution, $I(t)$, asymptotically vanishes with time.

\begin{theorem}\label{assymp-stability}
Let $\vec{u}$ be a solution of the linearized system \eqref{prey_ln}-\eqref{toxin_ln}. Assume that
the coefficients of the equations and depth of the reservoir $H$ are such that
\begin{align}\label{Du-chi-condb}
&Q_u> q_u,\\
&Q_v>q_v,\label{Dv-chi-condb}\\
&Q_c > q_c. \label{Dc-condb}
\end{align}
Then 
\begin{equation}\label{alpha-estimateb}
I(t)\leq I(0)\cdot e^{-\beta_0 t}.
\end{equation}
Here $\beta_0>0$ is a constant depending on  $q_s$ and $Q_s$.
\end{theorem}

\section{Numerical Simulation of the non-homogeneous time dependent equilibrium state}\label{sims}
We conducted numerical simulations of system \eqref{prey-eq-1-mod}-\eqref{pred-eq-1-mod} assuming Dirichlet boundary conditions. As expected the equilibrium state corresponded with our analytical solution, since we proved the uniqueness and existence theorems (figures not shown).  Furthermore, we conducted numerical simulations of the system with no flux boundary conditions, see Fig. \ref{fig:sims}. We numerically explored the question whether the periodicity of our solution was driven by the boundary conditions, or if the coefficients alone, assuming no flux boundary conditions, results in periodicity. We were able to numerically show that a periodic pattern is preserved in time and space, even with no0flux boundary conditions, see Fig. \ref{fig:sims}. 
\begin{figure}[H] \label{fig:sims}
\begin{center}
\includegraphics[width=0.8\textwidth]{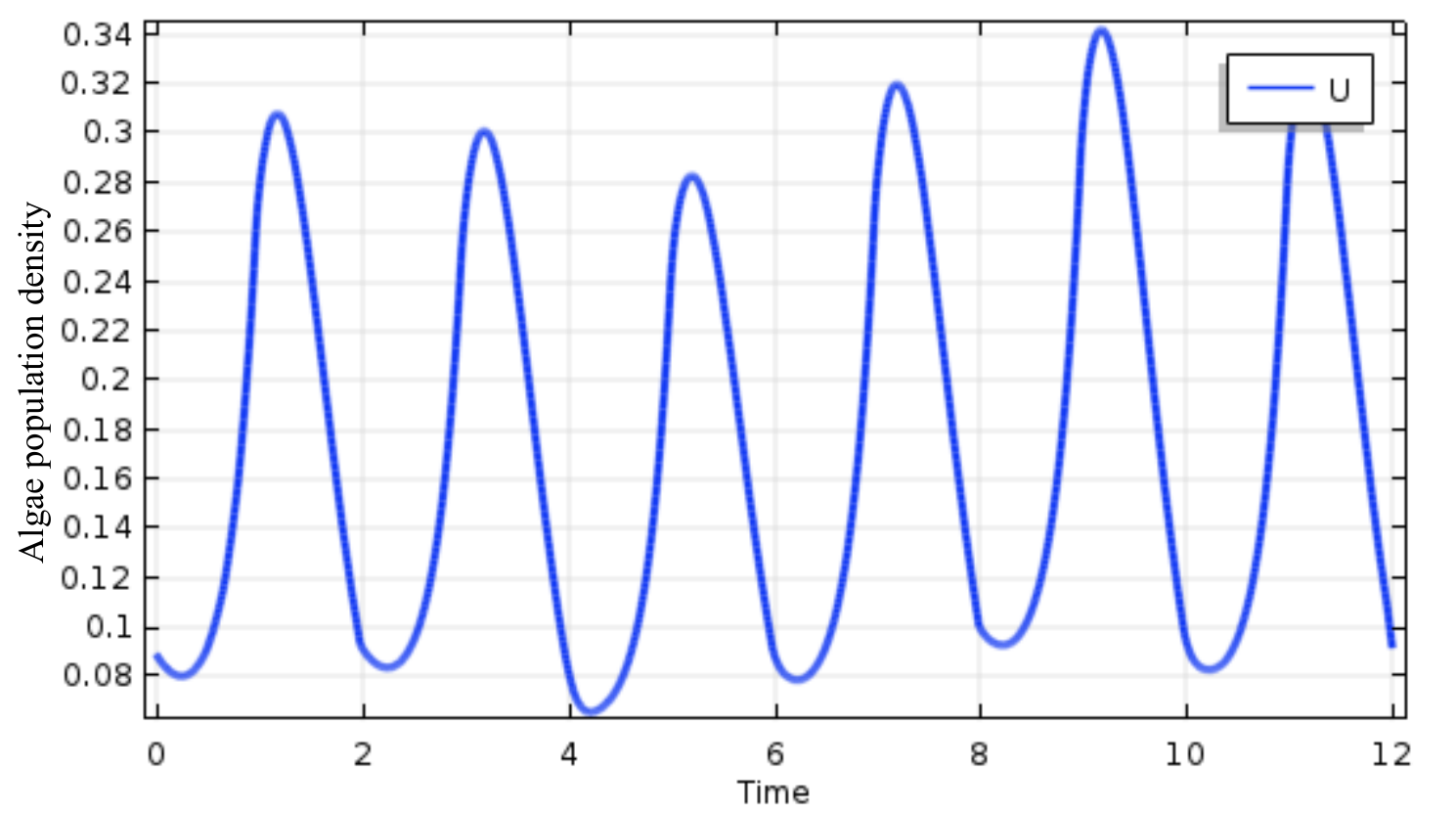}
\includegraphics[width=0.8\textwidth]{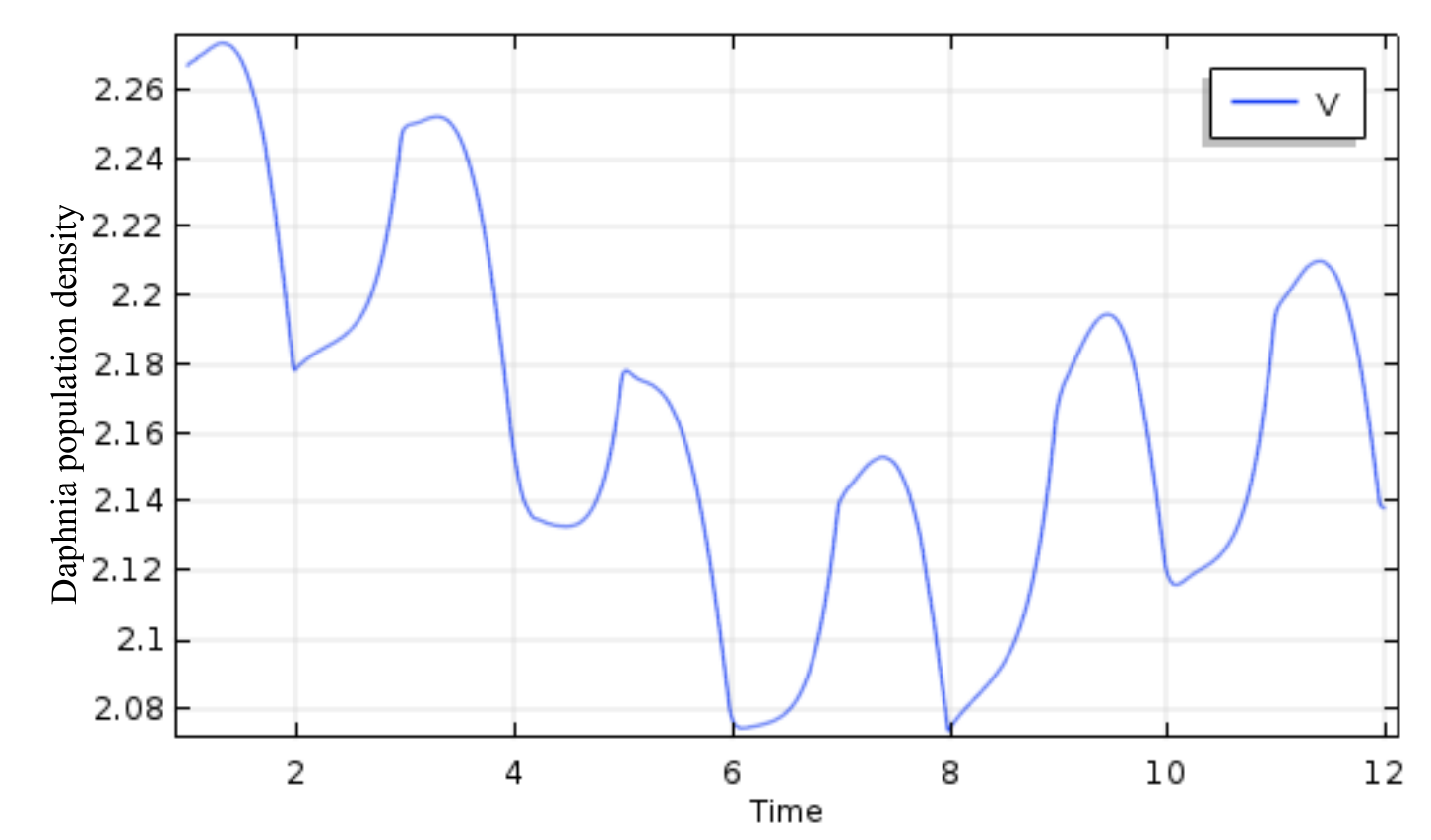}
\caption{Numerical solution of fully nonlinear equilibrium system $u_e$ Algae with no flux boundary condition of algae $u_e$ (top) and \textsl{Daphnia} $v_e$ (bottom). Simulations conducted in COMSOL.}
\end{center}
\end{figure}

\section{Discussion}
These modeling efforts explore an aquatic predator-prey system subject to a toxicant contaminant, whose population dynamics are driven by light. We extended previous modeling frameworks that neglect spatial dynamics and assume constant light levels. Here, we assume mortality varies with light availability, and allow the mortality rate coefficient to be subjected to switching regimes for different time intervals, depending on whether light is increasing or decreasing. For simplicity in this paper, we assume light availability is increasing on time interval $[0,1]$ and decreasing on time interval $[1,2]$, however more complex light patterns can be incorporated into the modeling framework.  The equilibrium state depends on  $\omega t$ in the boundary condition. If we ignore temporal fluctuations in light levels, the parameters of our model become time independent.

We proved the existence of an equilibrium state under certain conditions on the parameters, however if we ignore spatial dependence in the production of algae, then the equilibrium state exists in much more general cases. 
It is also worth noting that in the case of spatially homogeneous coefficients (when $\gamma=0$), a more general domain can be considered with $\Gamma_1$ being the set of $\partial{\Omega}$ such that $\Gamma_2$ is a  $C^1$ graph and the 2-D measure of $\Gamma_1$ is positive. 

While an algae-\textsl{Daphnia} predator-prey system motivated our model formulation, the resulting model is generic and can be applied to a wide variety of problems whose spatial-dependent parameters also depend on time, as periodic light driven coefficients.

\section{Conclusion}\label{sec:conclusions}
In this paper, we developed a model of light driven ecological process of the interactions between algae and \textsl{Daphnia}. We assume that some main parameters depend on time and reservoir. We proved the existence and uniqueness of spatial and temporal dependent periodic solutions (section \ref{uniqness}) and constructed analytical functions of the solutions given some constraints (section \ref{sec:analytical}). We investigated the Turing stability of the solution with respect to perturbations of initial conditions (section \ref{sec:stability}) and captured numerical simulations under more general boundary and initial conditions (section \ref{sims}). Given a perturbation to the equilibrium state solution, we showed the system will return to the equilibrium state solution as long as the motility coefficients are large enough and/or the reservoir depth is shallow enough.   
Additionally, it appears that our periodic solutions are not driven by the Dirichlet boundary conditions taken to facilitate analytical analyses, as numerical simulations assuming no flux boundary conditions exhibit similar periodic solutions.

%
%


\section*{Acknowledgments}
The authors would like to acknowledge the assistance of Dr. Eugenio Aulisa for his helpful guidance on the implementation of COMSOL used for numerical simulations. 
Author AP was partially supported by NSF grant DMS-1615697. 
\bibliographystyle{siamplain}
\bibliography{References.bbl}
\end{document}